\documentclass[10pt,journal,compsoc]{IEEEtran}
\ifCLASSOPTIONcompsoc
  \usepackage[nocompress]{cite}
\else
  \usepackage{cite}
\fi
\ifCLASSINFOpdf
\else
\fi
\usepackage{amsmath,amssymb}
\usepackage{changepage}
\usepackage{tikz}
\usetikzlibrary{positioning, calc}
\usetikzlibrary{shapes.geometric}
\usepackage[flushleft]{threeparttable}
\usepackage{array,booktabs,makecell}
\usepackage{setspace}
\usepackage{algorithm}
\usepackage[noend]{algpseudocode}
\usepackage{algorithmicx}
\usepackage{graphicx}

\makeatletter
\newcommand{\subalign}[1]{%
  \vcenter{%
    \Let@ \restore@math@cr \default@tag
    \baselineskip\fontdimen10 \scriptfont\tw@
    \advance\baselineskip\fontdimen12 \scriptfont\tw@
    \lineskip\thr@@\fontdimen8 \scriptfont\thr@@
    \lineskiplimit\lineskip
    \ialign{\hfil$\m@th\scriptstyle##$&$\m@th\scriptstyle{}##$\hfil\crcr
      #1\crcr
    }%
  }%
}
\makeatother

\usepackage{physics}
\usepackage{cancel}
\newcolumntype{C}[1]{>{\centering\arraybackslash}p{#1}}
\newcolumntype{L}[1]{>{\raggedright\arraybackslash}p{#1}}

\usepackage{amsthm}
\usepackage{thmtools}
\declaretheoremstyle[
spaceabove=6pt, spacebelow=6pt,
headfont=\normalfont\bfseries,
notefont=\mdseries, notebraces={(}{)},
bodyfont=\normalfont,
postheadspace=0.6em,
headpunct=:
]{mystyle}

\usepackage{url}
\usepackage{tabularx}
\makeatletter
\newcommand{\multiline}[1]{%
  \begin{tabularx}{\dimexpr\linewidth-\ALG@thistlm}[t]{@{}X@{}}
    #1
  \end{tabularx}
}

\usepackage{footnote}
\usepackage{adjustbox}

\makeatletter
\renewcommand{\Function}[2]{%
  \csname ALG@cmd@\ALG@L @Function\endcsname{#1}{#2}%
  \def\jayden@currentfunction{#1}%
}
\newcommand{\funclabel}[1]{%
  \@bsphack
  \protected@write\@auxout{}{%
    \string\newlabel{#1}{{\jayden@currentfunction}{\thepage}}%
  }%
  \@esphack
}
\makeatother

\usepackage{graphicx}
\usepackage{xcolor}

\newcommand{\distance}{{\textrm{d}}}
\newcommand{\laplace}{{\tilde{\mathcal{L}}}}
\newcommand{\lnear}{{\tilde{L}}}

\newcommand{\Ic}{\mathcal{I}_c}
\DeclareMathOperator*{\argmin}{arg\,min}
\DeclareMathOperator*{\argmax}{arg\,max}
\usepackage{array}
\usepackage{multirow}
\usepackage{amsfonts}
\usepackage{enumerate}
\usepackage{amsthm}
\usepackage{svg}
\usepackage{bbm}
\makeatletter
\algnewcommand{\LineComment}[1]{\Statex \hskip\ALG@thistlm \(\triangleright\) #1}
\makeatother

\newtheorem{thm}{Theorem}[section]

\newtheorem{prop}[thm]{Proposition}

\begin{document}
\newcommand{\equilateral}{
\begin{tikzpicture}[scale=0.8]
    \coordinate[]  (A) at (0,0);
    \coordinate[] (B) at (.4,0);
    \coordinate[] (C) at (.2,.3464);
    \fill [blue] (A) -- (B) -- (C) -- cycle;
  \end{tikzpicture}
}
%
\title{Fréchet Statistics Based Change Point Detection in Dynamic Social Networks} 

\author{Rui~Luo,
        Vikram~Krishnamurthy,~\IEEEmembership{Fellow,~IEEE}
\IEEEcompsocitemizethanks{\IEEEcompsocthanksitem R. Luo is with the Sibley School of Mechanical and Aerospace Engineering, Cornell University, Ithaca, NY, 14850.\protect\\
E-mail: rl828@cornell.edu
\IEEEcompsocthanksitem V. Krishnamurthy is with the School of Electrical and Computer Engineering, Cornell University, Ithaca, NY, 14850.\protect\\
E-mail: vikramk@cornell.edu
\IEEEcompsocthanksitem This research was supported in part by the  U. S. Army Research Office under grant W911NF-21-1-0093, and the National Science Foundation under grant CCF-2112457.\protect\\}}

\IEEEtitleabstractindextext{
\begin{abstract}
This paper proposes a method to detect change points in dynamic social networks using Fréchet statistics. We address two main questions: (1) what metric can quantify the distances between graph Laplacians in a dynamic network and enable efficient computation, and (2) how can the Fréchet statistics be extended to detect multiple change points while maintaining the significance level of the hypothesis test?
Our solution defines a metric space for graph Laplacians using the Log-Euclidean metric, enabling a closed-form formula for Fréchet mean and variance. We present a framework for change point detection using Fréchet statistics and extend it to multiple change points with binary segmentation. The proposed algorithm uses incremental computation for Fréchet mean and variance to improve efficiency and is validated on simulated and two real-world datasets, namely the UCI message dataset and the Enron email dataset.
\end{abstract}

}

\maketitle

\IEEEdisplaynontitleabstractindextext

%
\IEEEpeerreviewmaketitle

\section{Introduction}
\label{sec:introduction}
In recent years, the availability of social network data has increased the demand for statistical network analysis in areas such as socializing, information sharing, and collaborative work. Researchers are interested in analyzing the dynamics of social networks to gain insights into social phenomena and to predict future events. 

Detecting changes in social network structures is critical to identify emerging trends and patterns that can provide insight into social dynamics, including the emergence of new social groups, the spread of social influence, or changes in social behavior. Change point detection is a fundamental technique for community detection and identifying the formation of echo chambers \cite{luo2021echo}, which can amplify bias and increase misinformation in online social networks \cite{luo2022mitigating}. By integrating changes in temporal patterns into segregation models, we can better understand and ultimately mitigate the effects of echo chambers on online social networks. Change point information is also useful in social learning setups \cite{krishnamurthy2022dynamics}, where risk-averse agents adjust their estimation of a varying network state.

Existing change point detection approaches can be categorized into three main groups: non-parametric methods, parametric methods, and Bayesian methods. Non-parametric methods, such as cumulative sum (CUSUM) \cite{wang2020univariate} and sliding window methods \cite{aminikhanghahi2017survey}, do not assume any specific distribution for the data and can handle both abrupt and gradual changes. Parametric methods, such as autoregressive models and Markov models \cite{truong2020selective}, assume a specific distribution for the data and can achieve higher accuracy but require more computational resources. Bayesian methods, such as Bayesian change point analysis \cite{adams2007bayesian}, use probabilistic models for the prior of the change point and the observation likelihood to estimate the change points and can handle uncertainty and missing data.

However, detecting changes in social networks is a challenging problem due to the large size of social networks and their complex dynamics. Additionally, the non-Euclidean nature of networks presents a challenge because traditional statistical tools developed for scalar and vector data are inadequate. Therefore, there is a need for efficient and effective methods to detect changes in social networks. Fréchet mean and variance provide a method for calculating mean and variance for metric space-valued random variables, which allows us to examine statistical data for data items located in abstract spaces without algebraic structure and operations such as networks. By defining networks using graph Laplacians and computing their Fréchet mean and variance, we can establish their location and spread.

This paper focuses on estimating the locations of potentially multiple change points in dynamic social networks represented as a sequence of graph snapshots. To do so, we use Fréchet means and variances to construct a test statistic, which, under the null hypothesis of no change point, converges to a Brownian bridge process. This result is based on Theorem 1 in \cite{dubey2020frechet}. More specifically, we address the following questions:

\noindent
{\bf Q1. (Metric Choice)} Given the set of graph Laplacians corresponding to the graph snapshots of a dynamic social network, what is a suitable metric that quantifies their distances and enables efficient computation?

\noindent
{\bf Q2. (Multiple Change Point)} Assuming the Fréchet statistics constructed for the given metric space can detect a single change point in dynamic social networks, how can we extend it to the multiple change point setting while maintaining the significance level of the hypothesis test?

\vspace{0.1in}

\noindent
{\bf Main Results and Organization: }

\noindent (1) In Section \ref{sec:metric}, we define a metric space for graph Laplacians by finding the nearest symmetric positive definite matrix for each Laplacian, and then use the Log-Euclidean metric to measure their distances. This allows us to derive a closed-form formula for Fréchet mean and variance under this metric, which provides a more accurate and efficient way of measuring the distances between graphs. We also derive a closed-form formula for Fréchet mean and variance under this metric in Section \ref{subsec:metric of graph}.

\noindent (2) In Section \ref{sec:frechet}, we present a framework for change point detection using Fréchet statistics and generalize it to multiple change points with binary segmentation. Our primary result, Algorithm \ref{alg:bisect}, uses incremental computation for Fréchet mean and variance to improve computational efficiency.

\noindent (3) Finally, in Section \ref{sec: empirical}, we validate our proposed algorithm on simulated networks as well as real-world UCI message network and Enron email network. We confirm the algorithm's performance through ground truth change points in network simulation and real events for the real-world networks.

\vspace{0.1in}
\noindent
{\bf Related Works:}

\noindent (1) \textit{Metrics for Symmetric Positive (Semi-)Definite Matrices. } 
The space of symmetric positive definite (SPD) matrices, which is widely used in computer vision, medical imaging, and machine learning, has been the focus of much research on matrix metrics. The Frobenius metric, defined as $\delta_F(X, Y) = \|X-Y \|_{\textrm{F}} = \sqrt{\sum_{i=1}^{m} \sum_{j=1}^{n} (x_{ij}-y_{ij})^{2}}$,  is a commonly used metric for matrices, but it suffers from the swelling effect,  i.e., the determinant of the average is larger than any of the original determinants \cite{lin2019riemannian}. To overcome this issue, researchers have proposed various non-Euclidean metrics that represent and handle SPD matrices as elements of a differentiable Riemannian manifold.

Several metrics have been proposed for SPD matrices, including  the Riemannian metric $\delta_R(X, Y) = \| \log(Y^{-1/2} X Y^{-1/2}) \|_{\textrm{F}}$ \cite{bhatia2009positive}, the Log-Euclidean metric $\delta_{LE} (X, Y) = \|\log(X) - \log(Y) \|_{\textrm{F}}$ \cite{arsigny2007geometric}, and the Jensen-Bregman “log-det” based matrix divergence $\delta_J(X, Y) = \textrm{logdet}(\frac{A+B}{2}) -\frac{1}{2}\textrm{logdet}(AB)$\cite{cherian2012jensen}. 
Vemulapalli and Jacobs \cite{vemulapalli2015riemannian} summarized the properties of these metrics, including their invariance to various transformations, and identified the Riemannian and Log-Euclidean metrics as the most popular ones. 
While most of the metrics are defined for SPD matrices of fixed dimension, Lim et al. \cite{lim2019geometric} defined a geometric distance between a pair of SPD matrices of different dimensions, which is described in terms of ellipsoids. 

In contrast to SPD matrices, symmetric positive semi-definite (SPSD) matrices are singular and do not have matrix logarithms. There are two main approaches to define metrics on SPSD matrices. The first approach \cite{vandereycken2013riemannian, massart2020quotient} exploits $S_{+}(p, n)$, the manifold of rank-$p$ SPSDs of size $n$, which can be identified with the quotient manifold  $\mathbb{R}_{*}^{n\times p} / \mathcal{O}_p$, where $\mathbb{R}_{*}^{n\times p}$ is the set of full-rank $n \times p$ matrices and $\mathcal{O}_p$ is the orthogonal group of order $p$. Bonnabel and Sepulchre \cite{bonnabel2010riemannian} proposed a metric for  $S_{+}(p, n)$ that is invariant with respect to all transformations that preserve angles and derived the geometric mean. 
The second approach involves adding a regularization term to transform the SPSD matrix to a PSD or truncating the spectrum. 
Dodero et al. \cite{dodero2015kernel} regularized the graph Laplacian to become positive definite by adding a regularization term and used the Log-Euclidean metric for downstream classification tasks.  
Shnitzer et al. \cite{shnitzer2022log} truncated the full spectrum of diffusion operators to a fixed length and proved that the spectrum truncation preserves the lower bound of the Log-Euclidean metric.

\noindent (2) \textit{Fréchet Analysis of Graph Laplacians. } 
The Fréchet mean \cite{frechet1948elements} is a concept in statistics that provides a representative center of a set of data objects in a metric space. In the context of graph Laplacians, the Fréchet mean can be used to represent the central point in the space of graph Laplacians, which facilitates further analysis and comparisons. 

Zhou and Müller \cite{zhou2022network} recently addressed the network regression problem where the responses are graph Laplacians and the covariates are properties of interest, such as COVID-19 new cases for a public traffic network and subject age for a brain connectivity network. The approach is based on conditional Fréchet mean regression \cite{petersen2019frechet}, and it allows for modeling the relationship between the Laplacians and the covariates, as well as for predicting the Laplacians for new covariate values.
Dubey and Müller \cite{dubey2022modeling} developed a framework for analyzing networks by studying the Fréchet mean of a collection of graphs. The approach relies on the Frobenius metric as a distance measure, which quantifies the similarity between two graphs in terms of their topology. The authors use the sample average of the graph Laplacians as the Fréchet mean trajectory to summarize the topology evolution of the collection of graphs, which can be useful in applications such as traffic prediction, shape analysis, and network modeling. 
Ferguson and Meyer \cite{ferguson2022theoretical} proposed a pseudometric for graphs defined by the $l_2$ norm between the eigenvalues of the adjacency matrices. They also provide an algorithm to approximate the sample Fréchet mean of a set of undirected unweighted graphs with a fixed size using the pseudometric. The algorithm has low computational complexity and can be used to estimate the Fréchet mean for large datasets. The authors illustrate the usefulness of their approach by applying it to synthetic datasets such as the Barabasi-Albert graphs and small world graphs.

\noindent (3) \textit{Change Point Detection in Social Networks. } 
In the field of change point detection in social networks, Wang et al. \cite{wang2017fast} proposed an algorithm based on a Markov generative process to analyze graph snapshots of dynamic social networks. The algorithm was tested on real-world networks, including political voting networks, but it fails to account for long-term dependence and assumes rare changes. 
Masuda and Holme \cite{masuda2019detecting} used a graph distance measure and hierarchical clustering to detect and cluster evolving states in social temporal networks. Their approach assumes the entire network system is described by a single system state, which could be relaxed to multi-state setup for social networks with community structure. 
Zhao et al. \cite{zhao2019change} introduced a model-free change point detection method for dynamic social networks that uses neighborhood smoothing to estimate edge probabilities. However, the algorithm is not applicable to directed networks or networks with an evolving number of nodes. 
Grattarola et al. \cite{grattarola2019change} proposed a data-driven method for detecting changes in stationarity in a stream of attributed graphs. They used an adversarial autoencoder to embed graphs on constant-curvature manifolds, and employed the Fréchet mean to represent the average of networks. The geodesic distances between embedded graphs and the Fréchet mean were then used to identify potential changes. Although the proposed method is effective in detecting changes in stationarity, it is not applicable to multiple change point settings.

\section{Metric Space for Dynamic Networks}\label{sec:metric}
In this section, we present an approach that quantifies the distance between two networks in a dynamic network setup. We introduce a metric space for networks that allows us to measure the similarity between two networks based on their geometric properties. Specifically, we define the metric between two networks as the Log-Euclidean metric of the nearest symmetric positive definite (SPD) matrices of their respective graph Laplacians. Moreover, this metric admits a closed-form Fréchet mean, which allows us to characterize the average structure of a set of networks.

\subsection{Dynamic Network}
A dynamic network \cite{hulovatyy2016scout} is a sequence of graph snapshots $G^{(1)}, G^{(2)}, \cdots$, where each snapshot $G^{(t)} = (V^{(t)}, E^{(t)})$ represents the undirected\footnote{In Section \ref{subsec:metric of graph}, we discuss how our method can be extended to handle directed graphs. Specifically, we describe how to identify the nearest symmetric positive definite (SPD) matrix for an asymmetric graph Laplacian.} weighted graph observed at discrete time $t$. We further restrict that $V^{(1)} =V^{(2)} =\cdots= V$, so that all the graph snapshots have the same set of vertices. To address the issue of varying node numbers in real-world dynamic network, one approach \cite{huang2020laplacian} is to add isolated nodes to the graph snapshots so that each one has the same number of nodes.

For a graph snapshot $G^{(t)}$ with $N=|V^{(t)}|$ nodes, the $(i, j)$ entry of the $N \times N$ adjacency matrix $a_{ij}$ represents the edge weight between nodes $i$ and $j$. 
The graph Laplacian $L^{(t)}$ is given by $L^{(t)} = D^{(t)} - A^{(t)}$, where $D^{(t)}$ is the diagonal degree matrix whose diagonal entries are the sum of the weights of the edges incident to each node, $d^{(t)}_{ii} = \sum_{j=1}^{N} a^{(t)}_{ij}$. The graph Laplacians determine the network uniquely.

\subsection{Metric over the Space of Graph Laplacians}\label{subsec:metric of graph}
The graph Laplacian $L^{(t)}$ is positive semi-definite, and its rank equals the number of nodes minus the number of communities, i.e., disconnected components in the graph. In \cite{ginestet2017hypothesis}, Ginestet et al. characterize the set of $N\times N$ graph Laplacians with rank $l$ as a submanifold of $\mathbb{R}^{N^2}$ of dimension $Nl -
l(l+1)/2$. 

The singularity of the graph Laplacian restricts the application of SPD metrics, such as the Log-Euclidean metric. To avoid this issue, we adopt an algorithm \cite{cheng1998modified} to locate the nearest symmetric positive definite (SPD) matrix to $L^{(t)}$ in the Frobenius norm. This algorithm is based on the SVD of $L^{(t)}$ and is numerically stable and efficient. The resulting SPD, $\tilde{L}^{(t)}$, is then used in place of $L^{(t)}$ in defining the metric space. For directed graphs, Theorem 2 in \cite{halmos1972positive} provides the 2-norm distance between a matrix (which may not be symmetric) and the nearest symmetric positive semidefinite (SPSD) matrix. We then use the algorithm \cite{cheng1998modified} to identify the nearest SPD of the original graph Laplacian.

We define a metric space $(\laplace, \distance)$ for the graph snapshots. $\laplace$ is the set of the nearest SPD matrices of graph Laplacians, and $\distance$ is a function ${\displaystyle \distance\,\colon \laplace\times \laplace\to \mathbb {R}_{+} }$. 
We define $\distance$ using the Log-Euclidean metric $\delta_{LE}(X, Y) = \|\log(X) - \log(Y)\|_{\textrm{F}}$. 
The Log-Euclidean metric is defined as a bi-invariant metric on the Lie group \cite{bourbaki2008lie} of SPD matrices, which is viewed as the classical Euclidean metric on the vector space \cite{arsigny2007geometric}. 

Assume that $\lnear^{(1)}, \cdots, \lnear^{(n)} \sim F$ are independent and identically distributed random variables in $(\laplace, \distance)$, it admits a closed-form Fréchet mean which is unique (See Theorem 3.13 in \cite{arsigny2007geometric}):
\begin{equation}\label{eq:closed form mean}
    \mu_F = \exp \left( \mathbb{E} \left(\log \left(\lnear \right) \right) \right), \; \hat{\mu}_F = \exp(\frac{1}{n} \sum_{i=1}^{n} \log(\lnear^{(i)})),
\end{equation}
where $\exp$ is the matrix exponential, and $\mu_F$, $\hat{\mu}_F$ are the population Fréchet mean, sample Fréchet mean, respectively. 
The sample Fréchet mean $\hat{\mu}_F$ is asymptotically consistent due to its existence and uniqueness \cite{petersen2019frechet}.

The Fréchet variance quantifies the spread of the random variable around its Fréchet mean. The population Fréchet variance and its sample version for $\lnear^{(1)}, \cdots, \lnear^{(n)} \sim F$ are:
\begin{equation}\label{eq:closed form variance}
\begin{split}
    V_F &= \mathbb{E}\left( \distance^2 (\mu_F, \lnear ) \right), \\
    \hat{V}_F &= \frac{1}{n} \sum_{i=1}^{n} \distance^2(\hat{\mu}_F, \lnear^{(i)}) \\
    & = \frac{1}{n} \sum_{i=1}^{n} \|\frac{1}{n} \sum_{j=1}^{n} \log(\lnear^{(j)}) - \log(\lnear^{(i)}) \|_{\textrm{F}}^2  \\
\end{split}
\end{equation}

The following proposition established the Central Limit Theorem for the sample Fréchet variance $\hat{V}_F$ in the metric space $(\laplace, \distance)$ under certain assumptions (see Assumptions 1-3 in \cite{dubey2019frechet}, which relate to the existence and uniqueness of $\hat{\mu}_F$, the complexity bound of the metric space, and the finiteness of the entropy integral of the metric space):

\begin{prop}[Central limit theorem for the Fréchet variance]\label{prop:CLT variance}
    Under certain assumptions (Assumptions 1-3 in \cite{dubey2019frechet}), 
    \begin{equation}\label{eq:CLT variance}
    \begin{split}
        n^{1/2} (\hat{V}_F - V_F) \rightarrow N(0, \sigma^2_F) \; \textrm{in distribution,}
    \end{split}
    \end{equation}
where $\sigma^2_F = \textrm{Var}\{\distance^2(\mu_F, \lnear) \}$.
\end{prop}
\begin{proof}
    See \cite{dubey2019frechet}.
\end{proof}

\section{Fréchet Statistcs-based Change Point Detection}\label{sec:frechet}
This section addresses network change point detection, which involves identifying changes in the statistical properties of a sequence of graphs. The focus is on the offline (retrospective) change point detection setup, where an ordered sequence of observations is available. The primary objective is to estimate both the location and number of change points, which can then be used to segment the dataset into different regimes, under the assumption that the data within each regime comes from some common underlying distribution.

\subsection{Estimating the Location of a Change Point}\label{subsec:one change point}
Let us consider an independent time-ordered sequence $\{Y^{(i)}\}_{i=1}^{n}$ that takes values in a metric space $(\laplace, \distance)$ defined in Section \ref{subsec:metric of graph}. In the simplest case, we hypothesize that there is at most one change point location, denoted by $0<\tau<1$. Specifically, $Y^{(1)}, \cdots, Y^{(\lfloor n\tau \rfloor)} \sim F_1$ and $Y^{(\lfloor n\tau \rfloor +1)}, \cdots, Y^{(n)} \sim F_2$, where $F_1$ and $F_2$ are unknown probability measures on $(\laplace, \distance)$ and $\lfloor x \rfloor$ is the greatest integer less than or equal to $x$. In this context, the aim is to test the null hypothesis of distribution homogeneity, denoted by $H_0: F_1 = F_2$, against the alternative hypothesis of a single change point, denoted by $H_1: F_1 \neq F_2$. 

In change point detection, it is often necessary to ensure that each segment of the observed sequence is of sufficient size to accurately represent its underlying statistical properties, such as the Fréchet mean and variance. To achieve this, we assume that the hypothesized change point location $\tau$ lies within a compact interval $\Ic=[c, 1-c] \subset [0, 1]$, for some positive constant $c$. Alternatively, other types of segment size constraints can be imposed, such as the minimum cluster size \cite{matteson2014nonparametric}, which specifies a priori the minimum number of consecutive observations required to form a segment.

To characterize the statistical properties of data from two segments separated by $u \in \Ic$, we compute the sample Fréchet mean of the segment consisting of observations before and after $\lfloor nu \rfloor$ as
\begin{equation}
\begin{split}
    \hat{\mu}_{[0, u]} &= \argmin\limits_{l \in \laplace} \frac{1}{\lfloor n\tau \rfloor} \sum\limits_{i=1}^{\lfloor n\tau \rfloor} \distance^2(Y^{(i)}, l),\\
    \hat{\mu}_{[u, 1]} &= \argmin\limits_{l \in \laplace} \frac{1}{n-\lfloor n\tau \rfloor} \sum\limits_{i=\lfloor n\tau \rfloor +1}^{n} \distance^2(Y^{(i)}, l), \nonumber
\end{split}
\end{equation}
and the corresponding sample Fréchet variance are 
\begin{equation}\label{eq:variance}
\begin{split}
    \hat{V}_{[0, u]} &=  \frac{1}{\lfloor n\tau \rfloor} \sum\limits_{i=1}^{\lfloor n\tau \rfloor} \distance^2(Y^{(i)}, \hat{\mu}_{[0, u]}),\\
    \hat{V}_{[u, 1]} &= \frac{1}{n-\lfloor n\tau \rfloor} \sum\limits_{i=\lfloor n\tau \rfloor +1}^{n} \distance^2(Y^{(i)}, \hat{\mu}_{[u, 1]}), 
\end{split}
\end{equation}

The contaminated version of Fréchet variances can be obtained by replacing the Fréchet mean of a segment with the mean of the complementary segment. This leads to the definitions: 
\begin{equation}
\begin{split}
    \hat{V}^C_{[0, u]} &=  \frac{1}{\lfloor n\tau \rfloor} \sum\limits_{i=1}^{\lfloor n\tau \rfloor} \distance^2(Y^{(i)}, \hat{\mu}_{[u, 1]}),\\
    \hat{V}^C_{[u, 1]} &= \frac{1}{n-\lfloor n\tau \rfloor} \sum\limits_{i=\lfloor n\tau \rfloor +1}^{n} \distance^2(Y^{(i)}, \hat{\mu}_{[0, u]}), \nonumber
\end{split}
\end{equation}
which are guaranteed to be at least as large as the correct version (\ref{eq:variance}). The differences $\hat{V}^C_1 - \hat{V}_{[0, u]}$ and $\hat{V}^C_2 - \hat{V}_{[u, 1]}$ can be interpreted as measures of the between-group variance of the two segments. 

Suppose we fix some $u \in \Ic$. As a result of the central limit theorem for Fréchet variances (Proposition \ref{prop:CLT variance}), the statistic $\sqrt{u(1-u)}(\sqrt{n}/\sigma)(\hat{V}_{[0, u]} - \hat{V}_{[u, 1]})$ has an asymptotic standard normal distribution under the null hypothesis $H_0$. Here, $\sigma$ denotes the asymptotic variance of the empirical Fréchet variance. This result provides a powerful tool for statistical inference, allowing us to test hypotheses about differences in Fréchet variances between two segments of data.  A sample based estimator for $\sigma^2$ (\ref{eq:CLT variance}) is 
\begin{equation}
    \hat{\sigma}^2 = \frac{1}{n} \sum_{i=1}^{n} \distance^4(\hat{\mu}, \lnear_i) - \left( \frac{1}{n} \sum_{i=1}^{n} \distance^2(\hat{\mu}, \lnear_i) \right)^2, \nonumber
\end{equation}
which is consistent under $H_0$ \cite{dubey2019frechet}, and 
\begin{equation}
    \hat{\mu}=\argmin\limits_{l\in \laplace} \frac{1}{n} \sum\limits_{i=1}^{n} \distance^2(Y^{(i)}, l), \quad 
    \hat{V} = \frac{1}{n} \sum\limits_{i=1}^{n} \distance^2(Y^{(i)}, \hat{\mu}). \nonumber
\end{equation}

We adopt the test statistic proposed in \cite{dubey2020frechet}, which is capable of detecting differences in both Fréchet means and Fréchet variances of the distributions $F_1$ and $F_2$. 
\begin{equation}\label{eq:Tn(u)}
\begin{split}
    T_n(u) = & \frac{u(1-u)}{\hat{\sigma}^2} \Big[ (\hat{V}_{[0, u]} - \hat{V}_{[u, 1]})^2 \\ 
    & \phantom{---} + (\hat{V}_{[0, u]}^C - \hat{V}_{[0, u]} + \hat{V}_{[u, 1]}^C - \hat{V}_{[u, 1]})^2 \Big]
\end{split}
\end{equation}
The quantity $(\hat{V}_{[0, u]} - \hat{V}_{[u, 1]})^2$ provides a measure of the difference in Fréchet variances between two segments of data. Specifically, a larger value of $(\hat{V}_{[0, u]} - \hat{V}_{[u, 1]})^2$ indicates a greater difference in the variability of the data in the two segments. On the other hand, $(\hat{V}_{[0, u]}^C - \hat{V}_{[0, u]} + \hat{V}_{[u, 1]}^C - \hat{V}_{[u, 1]})^2$ captures the difference in Fréchet means between the two segments. 

Theorem 1 in \cite{dubey2020frechet} shows that under $H_0$, $\{nT_n(u): u\in \Ic \}$ converges weakly\footnote{Weak convergence is a function space generalization of convergence in distribution \cite{billingsley2013convergence}.} to the square of a standardized Brownian bridge on the interval $\Ic$, which is given by
\begin{equation}
    \mathcal{G} = \left\{ \frac{\mathcal{B}(u)}{\sqrt{u(1-u)}}: u\in \Ic \right\}, \nonumber
\end{equation}
where $\{\mathcal{B}(u): u\in \Ic \}$ is a Brownian bridge on $\Ic$, i.e., a Gaussian process indexed by $\Ic$ with zero mean and covariance structure given by $K(s, t) = \min(s,t) - st$. 
To perform a hypothesis test between $H_0$ and $H_1$, we use the statistic
\begin{equation} 
    \sup\limits_{u\in \Ic} nT_n(u) = \max\limits_{\lfloor nc \rfloor \leq k \leq n - \lfloor nc \rfloor} nT_n(\frac{k}{n}) \nonumber
\end{equation}
Here, $T_n(u)$ is a test statistic for the hypothesis test, which is computed for each potential change point $u\in \Ic$. 
To proceed, we need to obtain the $(1-\alpha)$th quantile of $\sup_{u\in \Ic} \mathcal{G}^2(u)$, denoted as $q_{1-\alpha}$. However, calculating this quantity by Monte Carlo simulations of $\mathcal{G}^2(\cdot)$ is inefficient or even infeasible when the dimension of the data is moderate to high. To overcome this, we employ a bootstrap approach, as described in Section 3.3 of \cite{dubey2020frechet}, to estimate $q_{1-\alpha}$.

Under $H_0$, the following  weak convergence holds:
\begin{equation}
    \sup\limits_{u\in \Ic} nT_n(u) \Rightarrow \sup\limits_{u\in \Ic} \mathcal{G}^2(u) \nonumber
\end{equation}
Using this, we define the rejection region for a level $\alpha$ significance test as:
\begin{equation}\label{eq:rejection region alpha}
    R_{n,\alpha} = \left\{ \sup\limits_{u\in \Ic} nT_n(u) > q_{1-\alpha} \right\}
\end{equation}

Under $H_1$, which assumes a change point is present at $\tau \in \Ic$, we can locate it by finding the maximizer of the process $T_n(u)$:
\begin{equation}
    \hat{\tau} = \argmax\limits_{u\in \Ic} T_n(u) = \argmax\limits_{\lfloor nc \rfloor \leq k \leq n - \lfloor nc \rfloor} T_n(\frac{k}{n})
\end{equation}
Here, $\hat{\tau}$ is the estimated change point, which maximizes the test statistic across all potential change points. 

By leveraging the closed form of the Fréchet mean (\ref{eq:closed form mean}) and variance (\ref{eq:closed form variance}), we have developed a recursive formula for updating these values incrementally, which is shown in the $\mathbf{\ref{func:incremental}}$ function of Algorithm~\ref{alg:bisect}. This approach significantly reduces the time complexity from $\mathcal{O}(n^2)$ to $\mathcal{O}(n)$, making it computationally efficient while still maintaining accuracy.

\subsection{Binary Segmentation for Estimating Multiple Change Points}
This subsection presents a binary segmentation procedure that extends the proposed statistic $T_n(u)$ (\ref{eq:Tn(u)}) to the multiple change point scenario. 
Binary segmentation is a computationally efficient tool that searches for multiple breakpoints in a recursive manner \cite{cho2012multi}. 
Assuming that $k-1$ change points have been estimated at locations $0 < \hat{\tau}_1 < \cdots < \hat{\tau}_{k-1} < 1$, the data is divided into $k$ clusters $\widehat{C}_1, \cdots, \widehat{C}_k$. 
Each cluster $\widehat{C}_j$ consists of observations between $\lfloor n \hat{\tau}_{j-1} \rfloor + 1$ and $\lfloor n \hat{\tau}_j \rfloor$, where $\hat{\tau}_0=0$ and $\hat{\tau}_k=1$ for brevity.

In the binary segmenation algorithm, we estimate the $k$th change point by applying the single change point procedure to the observations within one cluster. Let us consider the $j$th cluster. To accomplish this, we use a statistic for the $\widehat{C}_j$, denoted by $T_{n_j}(u)$, which is defined as follows:
\begin{equation}
\begin{split}
    T_{n_j}(u) = & \frac{u(1-u)}{\hat{\sigma}^2} \Big[ (\hat{V}_{[\hat{\tau}_{j-1}, u]} - \hat{V}_{[u, \hat{\tau}_j]})^2 \\ 
    & \phantom{---} + (\hat{V}_{[\hat{\tau}_{j-1}, u]}^C - \hat{V}_{[\hat{\tau}_{j-1}, u]} + \hat{V}_{[u, \hat{\tau}_j]}^C - \hat{V}_{[u, \hat{\tau}_j]})^2 \Big]
\end{split}
\end{equation}
Now, the estimated location of the $k$th change point is
\begin{equation}
    \hat{\tau}_k = \argmax\limits_{\hat{\tau}_{j-1} < u < \hat{\tau}_j} T_n(u)    
\end{equation}
Similar to the segment size constraint imposed by $\Ic$ in the one change point setting (Section \ref{subsec:one change point}), we assume that $c \leq \frac{\hat{\tau}_k - \hat{\tau}_{j-1}}{\hat{\tau}_j - \hat{\tau}_{j-1}} \leq 1-c$.

If a change point is detected at $\hat{\tau}_k$, we split the $j$th cluster $\widehat{C}_j$ into two new clusters: $Y^{\lfloor n \hat{\tau}_{j-1} \rfloor + 1}, \cdots, Y^{\lfloor n \hat{\tau}_{k} \rfloor}$, and $Y^{\lfloor n \hat{\tau}_k \rfloor + 1}, \cdots, Y^{\lfloor n \hat{\tau}_{j} \rfloor}$. We then proceed to analyze the two new clusters separately. The procedure is summarized in Algorithm \ref{alg:bisect}.

\renewcommand*\footnoterule{}
\begin{savenotes}
\begin{algorithm}
\caption{\textbf{Fréchet Binary Segmentation for Multiple Change Point Estimation}}\label{alg:bisect}
\textbf{Input}: 
Significance level $\alpha$, minimum segment length parameter $c$, bootstrap sample size $B$, length of each bootstrap sample $m$, a sequence of nearest SPDs corresponding to each snapshot of the dynamic network $\tilde{\mathrm{L}}=\{\lnear^{(1)}, \cdots, \lnear^{(n)}\}$. \

\textbf{Output}: A set of detected change points $\hat{\mathrm{T}} = \{\hat{\tau}_1, \hat{\tau}_2, \cdots \}$.

\begin{algorithmic}[1]
\Function{IncrementalFrechetStatistics}{$\tilde{\mathrm{L}}$} \funclabel{func:incremental}
\State $\hat{\mu}_{[0, \frac{1}{n}]} = \lnear^{(1)}$, $\hat{\mu}_{[\frac{n-1}{n}, 1]} = \lnear^{(n)}$, $\hat{V}_{[0, \frac{1}{n}]}=\hat{V}_{[\frac{n-1}{n}, 1]}=0$.
\For{$t=2, \cdots, n$}\footnotemark{} 
    \State \multiline{$\hat{\mu}_{[0, \frac{t}{n}]} = \frac{t-1}{t} \hat{\mu}_{[0, \frac{t-1}{n}]} + \frac{1}{t} \lnear^{(t)}$, \\
    $\hat{V}_{[0, \frac{t}{n}]} = \frac{t-1}{t} \hat{V}_{[0, \frac{t-1}{n}]} + \frac{1}{t} (\lnear^{(t)} - \hat{\mu}_{[0, \frac{t-1}{n}]} ) : (\lnear^{(t)} - \hat{\mu}_{[0, \frac{t}{n}]} )$, ($:$ denotes the Frobenius product.) \\
    $\hat{\mu}_{[\frac{n-t}{n}, 1]} = \frac{t-1}{t} \hat{\mu}_{[\frac{n-t+1}{n}, 1]} + \frac{1}{t} \lnear^{(n-t+1)}$, \\
    $\hat{V}_{[\frac{n-t}{n}, 1]} = \frac{t-1}{t} \hat{V}_{[\frac{n-t}{n}, 1]} + \frac{1}{t} (\lnear^{(n-t+1)} - \hat{\mu}_{[\frac{n-t+1}{n}, 1]} ) : (\lnear^{(n-t+1)} - \hat{\mu}_{[\frac{n-t}{n}, 1]} )$
    }
\EndFor
\For{$t=1, \cdots, n$} 
    \State \multiline{$\hat{V}_{[0, \frac{t}{n}]}^C = \hat{V}_{[0, \frac{t}{n}]} + (\hat{\mu}_{[\frac{t}{n}, 1]} - \hat{\mu}_{[0, \frac{t}{n}]}) : (\hat{\mu}_{[\frac{t}{n}, 1]} - \hat{\mu}_{[0, \frac{t}{n}]})$, \\
    $\hat{V}_{[\frac{n-t}{n}, 1]}^C = \hat{V}_{[\frac{n-t}{n}, 1]} + (\hat{\mu}_{[\frac{t}{n}, 1]} - \hat{\mu}_{[0, \frac{t}{n}]}) : (\hat{\mu}_{[\frac{t}{n}, 1]} - \hat{\mu}_{[0, \frac{t}{n}]})$}
\EndFor
\State Compute $\{nT_n(\frac{t}{n})\}_{t=1, \cdots, n}$ according to (\ref{eq:Tn(u)}). 
\State \Return $\sup_{u\in \Ic} nT_n(u)$, $\argmax_{u\in \Ic} T_n(u)$
\EndFunction
\end{algorithmic}
\begin{algorithmic}[1]
\Function{BinarySegmentation}{$\tilde{\mathrm{L}}$} \funclabel{func:binary}
\State \multiline{$q_{1-\alpha}=$\Call{Bootstrap}{$B$, $m$, $\alpha$}\footnotemark{}}
\State \multiline{$z, \hat{\tau} =$ \Call{IncrementalFrechetStatistics}{$\tilde{\mathrm{L}}$}}
\If{$z > q_{1-\alpha}$}
    \State Update $\hat{\mathrm{T}} \gets \hat{\mathrm{T}} \cup \{\hat{\tau} \}$
    \State \Call{BinarySegmentation}{$\lnear^{(1)}, \cdots, \lnear^{(\lfloor n\hat{\tau} \rfloor)}$}
    \State \Call{BinarySegmentation}{$\lnear^{(\lfloor n\hat{\tau} \rfloor +1)}, \cdots, \lnear^{(n)}$}
\EndIf
\State \Return $\hat{\mathrm{T}}$
\EndFunction
\end{algorithmic}
\end{algorithm}
\end{savenotes}
\addtocounter{footnote}{-2}
\stepcounter{footnote}\footnotetext{To reduce the accumulation of numerical errors of $\hat{\mu}_{[\cdot, 1]}$ and $\hat{V}_{[\cdot, 1]}$, we compute their values in reverse order, starting from the last index and working backwards.}
\stepcounter{footnote}\footnotetext{We refer the reader to Section 3.3 of \cite{dubey2020frechet} for the bootstrap scheme.}

\section{Empirical Analysis on Simulated and Real-world Networks}\label{sec: empirical}
We evaluate the performance of the proposed Algorithm \ref{alg:bisect} on multiple change point detection on simulated networks and two real-world datasets. Our results are completely reproducible; the code and datasets used in the experiments are publicly available at \url{https://tinyurl.com/frechet-network}.

\begin{figure}
	\centering
	\includegraphics[width=0.5\textwidth]{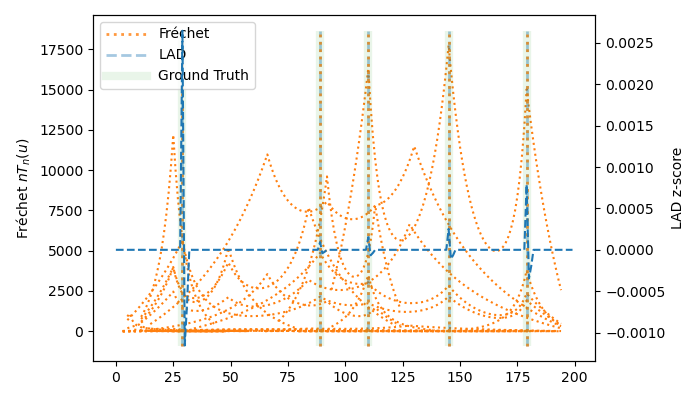}
	\caption{Comparison of our algorithm (Fréchet) with LAD for detecting multiple change points on a dynamic synthetic network of length 200. Both algorithms accurately detect all the change points, as confirmed by the ground truth. Test statistics of the two methods are also plotted: the Fréchet test statistic ${nT_n(u): u\in \Ic}$ for our algorithm and the z-score for LAD. The Fréchet test statistic for each analyzed segment is displayed since binary segmentation is used.}
	\label{fig: synthetic}
\end{figure}

\subsection{Baseline Method}
We use the Laplacian Anomaly Detection (LAD) \cite{huang2020laplacian} as baseline method. 
Our proposed change point detection algorithm uses Fréchet statistic of nearest SPDs of graph Laplacians. It is compared with the recent method LAD, which also employs graph Laplacians. 
LAD uses a low-dimensional embedding to summarize each snapshot of a dynamic network. This embedding is constructed from the singular vectors of the graph Laplacian, and captures temporal dependence by using summary vectors from two sliding windows of different lengths. To detect potential changes, LAD computes a z-score by comparing the current embedding with historical embeddings.

In contrast, our algorithm detects changes in the graph Laplacians themselves by computing the Fréchet distance between the nearest SPD matrices. This approach is more directly related to the underlying network structure and may be more effective in detecting changes in the network topology. Similar to LAD, our method can be adapted to handle changes in the number of nodes.

\begin{figure}
	\centering
	\includegraphics[width=0.5\textwidth]{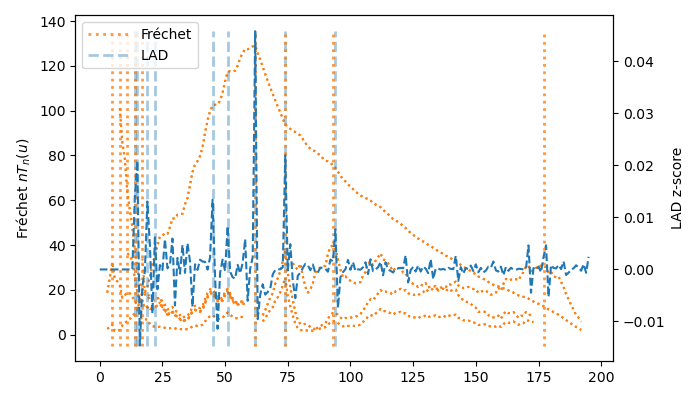}
	\caption{Comparison of our algorithm (Fréchet) with LAD for detecting multiple change points on the UCI message network. Both algorithms closely match the estimated change point locations, even though there is no ground truth available. }
	\label{fig: UCI}
\end{figure}

\subsection{Simulation Study}
To generate synthetic networks with change points, we utilize a network generation model based on the stochastic block model (SBM) \cite{huang2020laplacian}. The SBM incorporates a community vector that assigns each node to a specific block, and a block affinity matrix that specifies the probability of connections within and between blocks. By using the SBM, we can generate undirected, weighted networks with a fixed number of nodes and specified parameters.

This model accounts for two types of change points: the number of equally sized communities can change, as can the block affinity matrix. These changes can result in alterations to the network structure.

We demonstrate the efficacy of our algorithm and compare it with LAD for detecting multiple change points on a dynamic network of length 200. 
Figure \ref{fig: synthetic} shows that both algorithms accurately detect all the change points, as confirmed by the ground truth generated during network formation. We also plot the test statistics of the two methods: the Fréchet test statistic ${nT_n(u): u\in \Ic}$ (\ref{eq:Tn(u)}) for our algorithm and the z-score for LAD. Since we use binary segmentation, we display the Fréchet test statistic for each analyzed segment.

\subsection{Experiments on Real-world Networks}

\subsubsection{UCI Message Network}
The UCI Message dataset \cite{panzarasa2009patterns} deals with an online community of students at the University of California, Irvine. It represents a directed and weighted dynamic network, where each node corresponds to a student user and each edge indicates a message interaction from one user to another. The edge weight reflects the number of characters exchanged between the two users. A self-loop with a unit weight is added to each user at account creation. 
The dataset covers communication patterns from April to October 2004, spanning 196 days, with 1,899 users sending a total of 59,835 messages. We treat each day as an individual time point in our analysis.

For the UCI message network, we show the multiple change points results in Figure \ref{fig: UCI}. Although no ground truth change points are known, we notice that both algorithms have a significant test statistic at day 65, which is the end of spring term, and remain relatively low between the end of spring term and the start of fall term (day 70 to day 160). The estimated change point locations of the two algorithms closely match each other.

\subsubsection{Enron Email Network}
Enron Corporation, an American energy company, was involved in a high-profile accounting fraud scandal that led to its bankruptcy in 2001. Following its collapse, email data from Enron employees was made public \cite{priebe2005scan}. Our study examines the weekly email activity\footnote{We used the processed version which is available at \url{http://www.cis.jhu.edu/~parky/Enron/}.} between employees from November 1998 to June 2002, to determine if changes in email patterns reflect events leading to the company's downfall. We analyze the Enron email network, which comprises 184 email addresses, during the period from November 1998 to June 2002. We represent each unique email address as a node and the number of exchanged emails as the edge weight between the nodes. We treat each week as an individual time point, resulting in a sequence of length 184.

We confirm the detected change point with \cite{dubey2020frechet}. The proposed algorithm successfully locates the date August 23, 2000 (week 89), just before a significant event in the timeline of Enron when its stock prices hit an all-time high, which LAD fails to detect. Table 1 shows other potential change points based on important events. These results provide valuable insights into communication patterns and organizational behavior during real-life corporate transformations or even scandals.

\begin{figure}
	\centering
	\includegraphics[width=0.5\textwidth]{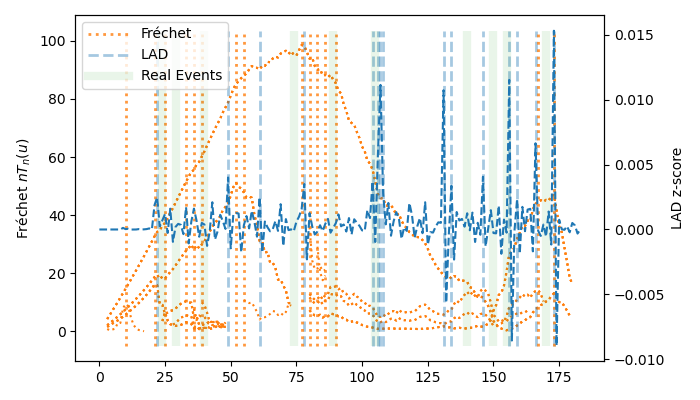}
	\caption{Comparison of our algorithm (Fréchet) with LAD for detecting multiple change points on the UCI message network. Both algorithms closely match the estimated change point locations, even though there is no ground truth available. }
	\label{fig: enron}
\end{figure}

\begin{center}\label{table:enron}
\begin{tabular}{ | m{5em} | m{0.6cm}| m{5.2cm} | } 
\hline
    Date & Week no. & Event  \\ \hline 
   05/24/1999& 24  & the Silverpeak Incident \\ \hline 
   06/28/1999&  29 & Board of Directors exempts CFO to run a private equity fund - LJM1. \\ \hline 
   09/16/1999&  40 & CFO addresses Merrill Lynch to find investors for LJM2 Fund. \\ \hline
   05/12/2000&  74 & Chief trader confirms strategy to exploit market via email. \\ \hline
   08/23/2000&  89 & Stock hits all-time high; FERC orders an investigations. \\ \hline
   12/13/2000&  105 & COO replaces CEO. \\ \hline
   08/13/2001&  140 & CEO resigns after board meeting. \\ \hline 
   10/22/2001&  150 & Enron acknowledges SEC inquiry. \\ \hline
   12/02/2001&  155 & Enron files for bankruptcy. \\ \hline
   03/14/2002&  170 & Former auditor indicted for obstruction of justice. \\
  \hline
\end{tabular}
\end{center}

\section{Conclusions and Extensions}
\label{sec:conclusion}
\noindent
{\bf Conclusions: }
This study addresses change point detection in dynamic social networks using the Fréchet mean and variance to locate and quantify changes in a sequence of graph Laplacians that correspond to graph snapshots. Our method builds upon the original work by Dubey and Müller \cite{dubey2020frechet} and extends it to multiple change point settings.
Compared to other methods, our approach offers several advantages. First, it leverages the closed-form expressions for the Fréchet mean and variance of SPDs, leading to efficient computation. Second, it uses an incremental update scheme, which further reduces computational complexity. Finally, it allows for fine-tuning of the detection accuracy through a user-defined significance level.

To evaluate the effectiveness of our proposed algorithm, we conducted numerical experiments using simulated networks and compared its performance to the LAD baseline method. Furthermore, we applied it to real-world networks. Our results show that our algorithm successfully identified real events in these networks. 
Overall, our experiments provide strong evidence for the efficacy and practical applicability of our proposed algorithm in analyzing dynamic social networks. 

\ifCLASSOPTIONcaptionsoff
  \newpage
\fi



%
\bstctlcite{IEEEexample:BSTcontrol}
\bibliographystyle{IEEEtran}




\end{document}